\documentclass[copyright,creativecommons]{eptcs}

\providecommand{\event}{GandALF 2016} % Name of the event you are submitting to

\usepackage{graphicx}
\usepackage{amsmath,amsthm}
\usepackage{xspace}
\usepackage{MnSymbol}
\usepackage{mathtools}
\usepackage{mathpartir}

\usepackage{tikz}
\usetikzlibrary{automata}

\newtheorem{theorem}{Theorem}
\newtheorem{lemma}[theorem]{Lemma}

\theoremstyle{definition}
\newtheorem{definition}[theorem]{Definition}

\newtheorem{example}{Example}
\newcommand{\true}{\ensuremath{\top}}
\newcommand{\false}{\ensuremath{\bot}}
\newcommand{\imp}{\ensuremath{\rightarrow}}
\newcommand{\biimp}{\ensuremath{\leftrightarrow}}
\newcommand{\Mudiam}[1]{\ensuremath{\langle #1 \rangle}}
\newcommand{\Mubox}[1]{\ensuremath{[ #1 ]}}

\newcommand{\mcs}{\textup{MCS}\xspace}

\newcommand{\pdl}{\textup{PDL}\xspace}
\newcommand{\pdltest}{\textup{PDL}\ensuremath{^{\scriptstyle ?}}\xspace}
\newcommand{\pdlcap}{\textup{PDL}\ensuremath{^\cap}\xspace}
\newcommand{\pdocap}{\texorpdfstring{\textup{PDL}\ensuremath{^{\cap,{\scriptstyle ?}}_{0}}\xspace}{Iteration-free PDL with Intersection}}
\newcommand{\pdtcap}{\textup{PDL}\ensuremath{^{\cap,{\scriptstyle ?}}}\xspace}
\newcommand{\arena}{\ensuremath{\operatorname{arena}}\xspace}
\newcommand{\dom}{\ensuremath{\operatorname{dom}}\xspace}

\newcommand{\Forw}{\textit{Forw}\xspace}
\newcommand{\Cyc}{\textit{Cyc}\xspace}
\newcommand{\forw}{\ensuremath{\mathsf{forw}}}
\newcommand{\cyc}{\ensuremath{\mathsf{cyc}}}

\newcommand{\lp}{\ensuremath{\mathit{lp}}}
\newcommand{\rp}{\ensuremath{\mathit{rp}}}

\def\newarrow#1{\mathop{{\hbox{\setbox0=\hbox{$\scriptstyle{#1\quad}$}{$%
\mathrel{\mathop{\setbox1=\hbox to
\wd0{\rightarrowfill}\ht1=3pt\dp1=-2pt\box1}\limits^{#1}}%
$}}}}}
\newcommand{\Transition}[3]{\ensuremath{#1 \newarrow{#2} #3}}

\title{A Canonical Model Construction for Iteration-Free PDL with Intersection}
\author{Florian Bruse \and Daniel Kernberger \and Martin Lange 
        \institute{School of Electr.\ Eng. and Computer Science, University of Kassel, Germany}}
\def\titlerunning{A Canonical Model Construction for \pdlcap}
\def\authorrunning{F.~Bruse, D.~Kernberger \& M.~Lange}
\begin{document}
\maketitle

%
% TODOs:
%
% - UK vs. US English
%

  \begin{abstract}
  We study the axiomatisability of the iteration-free fragment of Propositional Dynamic Logic with Intersection and Tests. The combination
  of program composition, intersection and tests makes its proof-theory rather difficult. We develop a normal form 
  for formulae which minimises the interaction between these operators, as well as a refined canonical model
  construction. From these we derive an axiom system and a proof of its strong completeness.  
  \end{abstract}

\section{Introduction}
% !TEX root =  pdl.tex

Propositional Dynamic Logic (\pdl) is a multi-modal logic with a two-sorted language. It defines \emph{formulae} 
and \emph{programs}. Formulae make assertions about worlds in a Kripke structure referring to the ability to 
access other worlds using programs. The term ``program'' originates from \pdl's early use in program specification 
and verification. It can be seen as a propositional Floyd-Hoare calculus \cite{pratt-1976:semanconsi:inbook:109}. 
\pdl is also related to logics used in knowledge representation, it is basically the same as the description logic 
$\mathcal{ALC}_{\mathsf{reg}}$ \cite{journals/jolli/PrendingerS96}. 

\pdl is a well-behaved modal logic in that its satisfiability problem is decidable. The first upper bound result
yielded NEXPTIME based on a small model property and polynomial time model checking \cite{Fischer79}. Later the
problem was shown to be EXPTIME-complete \cite{pratt-1980,Fischer79}.   
Axiomatisations for \pdl have also been given, for instance in form of the Segerberg axioms \cite{%Seg77,
Parikh78}
and others, c.f.\ \cite{KozPar81}.
%\cite{Pratt78,ACTAI::Nishimura1979,KozPar81,conf/lop/Vakarelov80,Knijnenburg91,lange-ki03}. 

In plain \pdl, the programs form a Kleene algebra, built from atomic programs with the operations \emph{union},
\emph{composition} and \emph{iteration}. A natural question concerns its extension with other program constructs
and its effect on expressive power, decidability, complexity, axiomatisability, etc. 
One such operator is \emph{test} which creates new basic programs from formulae (hence making them mutually
recursive). \pdl with Tests (\pdltest) is more expressive than \pdl \cite{Berman-Paterson/81} but satisfiability
has the same complexity. Thus, adding just tests does not create conceptual problems; it only tends to complicate 
correctness proofs slightly. 
Other operators that have been found to be equally harmless in this respect are \emph{looping} \cite{IC::Streett1982} 
and \emph{converse} \cite{Fischer79}. 
%The looping constructs allows assertions about infinite repetitions of program executions
%to be made, and the converse operator allows a program to run backwards. Adding either or both does not
%make the logic undecidable \cite{IC::Streett1982}, in fact the complexity of satisfiability remains EXPTIME 
%which can be seen using a translation into the modal $\mu$-calculus with converse modalities for instance
%\cite{conf/icalp/Vardi98}.

Another program construct that one may consider naturally is \emph{intersection}. \pdl with Intersection (\pdlcap)
turns out to be more complex than \pdl. Its satisfiability problem is -- perhaps surprisingly -- decidable, but it 
is 2EXPTIME-complete \cite{lncs208*34,langelutz-jsl05}. The addition of intersection causes the loss of 
bisimulation-invariance and, hence, the tree model property, but it preserves a DAG model property. This does not 
hold true in the presence of tests anymore, though. \pdtcap can require models to have (nested!) cycles. 

The intersection operator also turns out to be intriguing for the problem of axiomatisability. There are studies
concerning sound and complete axiomatisiations for versions of \pdl with Intersection. Balbiani and
Fari\~{n}as del Cerro axiomatised the small fragment without tests, Kleene stars and union \cite{journals/jancl/BalbianiC98}. 
Passy and Tinchev have shown that \pdtcap without the Kleene star (\pdocap) is axiomatisable in the stronger language of added nominals 
\cite{passy91ic}. Nominals can be seen as atomic propositions that hold true in a unique world of a Kripke structure. 
At close inspection one can see that the concept of naming particular worlds -- which is known to break 
bisimulation-invariance as well -- is very helpful on the way to a sound and complete axiomatisation for \pdl with 
Intersection: as we will see in Sect.~\ref{sec:canonical}, the intersection operator can be used to require several 
copies of worlds that should otherwise look exactly the same. Nominals for instance can help to distinguish these copies.

Surely, an axiomatisation of a logic that uses features which are not available in the logic can be questioned. Balbiani 
and Vakarelov have re-considered the problem of axiomatisation for \pdocap and proposed a deductive system
which only use logical \emph{operators} that are available in the logic \cite{DBLP:journals/fuin/BalbianiV01,Ba2003.6}. 
However, the axiomatisation requires a larger vocabulary than the logic to be axiomatised has, see Sect.~\ref{sec:canonical}
for details. In \cite{DBLP:journals/jancl/BalbianiV03}, Balbiani and Vakarelov extended their work to full \pdtcap, but
the issues with the differences between the object logic and the proof logic persist. The work that is closest to the one
presented here is Balbiani's refinement of the \pdocap axiomatisation \cite{Ba2003.6}. That calculus does not 
seem to rely on additional features outside of the object logic. Working out the exact connection between that axiomatisation
and the one presented here is left as future work, see also the concluding remarks at the end of this paper.

Here we propose a normal form for \pdocap formulae which minimises the interaction of intersections and
tests in its programs. We then present a canonical model construction for \pdocap which starts with all maximally consistent
sets and introduces abstract accessibility relations between them that correspond to non-atomic programs. These abstract
edges then get refined, possibly introducing new copies of maximally consistent sets until the model is saturated. 
Then every maximally consistent set is satisfied somewhere in this limit model, so we can derive a complete
axiomatisation for \pdocap over any set of propositions $\mathcal{P}$ which does not need additional operators or propositions.

%%% Local Variables:
%%% mode: latex
%%% TeX-master: "pdl"
%%% End:

\section{Iteration-Free PDL with Intersection and Tests}
% !TEX root =  pdl.tex

We fix a set $\mathcal{P} = \{ p, q,\ldots \}$ of unary relation symbols called \emph{atomic propositions} and a set 
$\mathcal{R} = \{a,b,\ldots \}$ of binary relation symbols called \emph{atomic programs} for the rest of the paper. We refer to 
$\mathcal{P}, \mathcal{R}$ as the vocabulary $\tau$.

Formulae $\varphi$ and programs $\alpha$ of Iteration-Free Propositional Dynamic Logic with Intersection and Tests (\pdocap) 
are derived by 
\[
	\varphi \enspace \coloneqq \enspace p \mid \varphi \vee \varphi \mid \neg \varphi \mid \langle \alpha \rangle \varphi \qquad\quad
	\alpha \enspace \coloneqq \enspace a \mid \alpha;\alpha \mid \alpha \cup \alpha \mid \alpha \cap \alpha \mid \varphi?
\]
where $p \in \mathcal{P}, a\in \mathcal{R}$. Other Boolean and modal connectives like $\true,\false,\wedge,\imp$, $\biimp$, $[\alpha]$ can be 
derived as usual. For ease of notation and for reasons of readability, we sometimes use the abbreviation $\alpha^\circlearrowleft \coloneqq \alpha \cap \true?$. We use $\wedge$, $\vee$, $\imp$, $\biimp$ as the descending order of precedence in formulae and $;$, $\cap$, $\cup$
in programs in order to save parentheses. Unary operators always bind stronger than binary ones.

%\paragraph{Kripke Structures.}
%\paragraph{Semantics.} 
Formulae are interpreted in worlds of a Kripke structure $\mathfrak{A}$; programs are interpreted as binary relations in these structures. A Kripke structure $\mathfrak{A}$ over $\tau$ is a set of worlds $W$ together with interpretations 
$P^{\mathfrak{A}}_p \subseteq A$ for all $p \in \mathcal{P}$ and $R^{\mathfrak{A}}_a \subseteq A \times A$ for all $a \in \mathcal{R}$. A pointed Kripke structure with distinguished world $u$ is written as $\mathfrak{A},u$. The semantics is given inductively as follows.
\begin{displaymath}
\begin{array}{rlclrll}
\mathfrak{A},u &\models p  &\text{ iff } & u \in P^{\mathfrak{A}}_p, p \in \mathcal{P} 
  &  R_{\varphi?}^{\mathfrak{A}} &\coloneqq & \{ (u,u) \mid \mathfrak{A},u \models \varphi \} \\
\mathfrak{A},u &\models \varphi \vee \varphi'  &\text{ iff } & \mathfrak{A},u \models \varphi \text{ or } \mathfrak{A},u\models \varphi'
  & R_{\alpha \cup \beta}^{\mathfrak{A}} &\coloneqq & R_{\alpha}^{\mathfrak{A}} \cup R_{\beta}^{\mathfrak{A}} \\
\mathfrak{A},u &\models \neg \varphi  &\text{ iff } & \mathfrak{A},u \not\models \varphi
  & R_{\alpha \cap \beta}^{\mathfrak{A}} &\coloneqq & R_{\alpha}^{\mathfrak{A}} \cap R_{\beta}^{\mathfrak{A}} \\
\mathfrak{A},u &\models \Mudiam{\alpha}\varphi  &\text{ iff } & \exists v. (u,v) \in R^{\mathfrak{A}}_{\alpha} \text{ and } \mathfrak{A},v \models \varphi \enspace
  & R_{\alpha;\beta}^{\mathfrak{A}} &\coloneqq & R_{\alpha}^{\mathfrak{A}} \circ R_{\beta}^{\mathfrak{A}}.

%\mathfrak{A},u &\models [\alpha] \varphi  & &\text{ iff }   \mathfrak{A},v \models \varphi \text{ f.a. } v \in \alpha[a] \\
\end{array}
\end{displaymath}
We say that $\varphi$ and $\psi$ are equivalent and write $\varphi \equiv \psi$ if, for all pointed Kripke structures $\mathfrak{A},u$, we have $\mathfrak{A},u \models \varphi$ if and only if $\mathfrak{A},u \models \psi$.

We write $\mathfrak{A},u \models \Phi$ if $\mathfrak{A},u \models \varphi$ for all $\varphi \in \Phi$.
We write 
%\Transition{u}{a}{v} to indicate that $(u,v) \in R^{\mathfrak{A}}_a$ and we write 
$\Transition{u}{\alpha}{v}$ to indicate that $(u,v) \in R^{\mathfrak{A}}_{\alpha}$. We write further $\alpha \Rightarrow \beta$ if for all Kripke structures $\mathfrak{A}$ and all worlds $u,v$ it holds that if $\Transition{u}{\alpha}{v}$, then it also holds that $\Transition{u}{\beta}{v}$.

\begin{definition}
\label{def:wg}
Let $\alpha$ be a program and let $\mathfrak{A}$ be a Kripke structure in which $\Transition{u}{\alpha}{v}$ for some $u,v$. A \emph{witness graph} for $\Transition{u}{\alpha}{v}$ is defined inductively over $\alpha$:
\begin{itemize}
\item If $\alpha = a$ for $a\in\mathcal{R}$ then a witness graph consists of nodes $u$ and $v$ together with the $a$-edge between $u$ and $v$.
\item If $\alpha = \varphi?$ then a witness graph is the node $u (=v)$.
\item If $\alpha = \alpha_1;\alpha_2$ then there is $w$ such that $\Transition{u}{\alpha_1}{w}$ and $\Transition{w}{\alpha_2}{v}$. The witness graph for $\Transition{u}{\alpha}{v}$ is the union of the witness graphs for $\Transition{u}{\alpha_1}{w}$ and $\Transition{w}{\alpha_2}{v}$.
\item If $\alpha = \alpha_1 \cap \alpha_2$ then there are witness graphs for $\Transition{u}{\alpha_1}{v}$ and $\Transition{u}{\alpha_2}{v}$, and their union is the witness graph for $\Transition{u}{\alpha}{v}$.
\item If $\alpha = \alpha_1 \cup \alpha_2$ then there is a witness graphs for $\Transition{u}{\alpha_1}{v}$ or for $\Transition{u}{\alpha_2}{v}$. Either of these is a witness graph for $\Transition{u}{\alpha}{v}$.
\end{itemize}
Witness graphs need not be unique. If \Transition{u}{\alpha}{v}, there can be many witness graphs in a particular Kripke structure. % This set of witness graphs is partially ordered by the subgraph relation. A minimal witness graph for $v \in \alpha[u]$ is one that is minimal with respect to this ordering.
\end{definition}

%\TODO{Definiere Modellbeziehung f\"ur Formel\emph{mengen} irgendwo}
%%% Local Variables:
%%% mode: latex
%%% TeX-master: "pdl"
%%% End:

\section{Canonical Models}
%:
\label{sec:canonical}
% !TEX root =  pdl.tex

A \emph{canonical model} (for a modal logic) is typically a Kripke structure whose worlds are all the maximally consistent sets
of the underlying logic whith respect to some notion of provability $\vdash$. Such a set $\Phi$ of formulae is \emph{consistent} if
it is not possible to derive a contradiction from it, i.e.\ if $\Phi \not\vdash \false$. It is a \emph{maximally consistent set} 
(MCS) if it is consistent and maximal with respect to $\subseteq$, i.e.\ it is not possible to add any formula of the underlying
logic without making it inconsistent. Canonical models are typically used to prove (strong) completeness of the axiomatisation $\vdash$
along the following lines. Suppose $\Phi \not\vdash \false$, i.e.\ $\Phi$ is consistent, then it is included in an MCS $\Phi'$. 
Next one shows that for every world $v$ in the canonical model $\mathfrak{K}$ and every formula $\varphi$ of the underlying logic we have 
$\mathfrak{K},v \models \varphi$ iff $\varphi \in v$. So a consistent set, and particularly an MCS, is satisfiable in the canonical model. 
Equally, every valid set of formulae is provable, i.e.\ the axiomatisation is complete. 

\paragraph{Canonical models and \pdocap.}
A simple consequence of the standard understanding of a canonical model is the fact that no two worlds in it represent the same
MCS. It is important to understand that no such standard canonical model construction can be used to prove completeness of an
axiomatisation for \pdocap as the following example shows.

\begin{example}
\label{ex:isectproblem}
Let $\Phi$ be a satisfiable set of \pdocap formulae. Consider the set
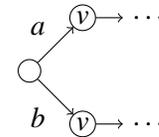
\begin{center}
\begin{minipage}{0.7\linewidth}
\[
\mathit{Split}(\Phi) \enspace \coloneqq \enspace \{ \Mudiam{a}\true, \Mudiam{b}\true, \Mubox{a \cap b}\false \} 
  \cup \{ \Mubox{a}\varphi, \Mubox{b}\varphi \mid \varphi \in \Phi \}\ . 
\]
\end{minipage} \hfill
\begin{minipage}{0.27\linewidth}
  \begin{tikzpicture}[every state/.style={minimum size=3mm, inner sep=1pt}]
    \node[state] (0)  {};
    \node[state] (1)  [above right of=0] {$v$};
    \node[state] (2)  [below right of=0] {$v$};
    \node        (1a) [right of=1, node distance=9mm] {\ldots};
    \node        (2a) [right of=2, node distance=9mm] {\ldots};
    
    \path[->] (0) edge node [above left] {$a$} (1)
                  edge node [below left] {$b$} (2)
              (1) edge                         (1a)
              (2) edge                         (2a);
  \end{tikzpicture}
\end{minipage}
\end{center}
It is easily seen to be satisfiable, too. Suppose $\mathfrak{A},v$ is a model of $\Phi$. A model for $\mathit{Split}(\Phi)$ is
obtained using two disjoint copies of $\mathfrak{A},v$ as shown on the right. It is equally possible to see that the two copies
cannot be merged since this would contradict the requirement $\Mubox{a \cap b}\false$. 
\end{example}

\paragraph*{Canonical model constructions for \pdocap in the literature.}
As mentioned in the introduction, the literature contains proposals for \pdocap axiomatisations, most notably by Balbiani and 
Vakarelov \cite{DBLP:journals/fuin/BalbianiV01,Ba2003.6}. The intricacies introduced by program intersection are tackled using
the following principle, c.f.\ \cite[Prop.~2.1]{DBLP:journals/fuin/BalbianiV01}: if $u$ has an $(\alpha \cap \beta)$-successor
then it has an $\alpha$-successor $v$ and a $\beta$-successor $w$ such that $v$ and $w$ cannot be distinguished by any atomic 
proposition $p$. It is important to note that $p$ is not restricted to be drawn from any pre-given set; instead it ranges over 
all propositions that could possibly \emph{extend} an underlying model. 

Balbiani and Vakarelov then formulate this semantic principle syntactically as a proof rule (INT) and present a refined canonical
model construction that circumvents the problem with intersection as outlined in Ex.~\ref{ex:isectproblem} as follows. Worlds
of the canonical model are not MCS but \emph{maximally consistent theories} (MCT). An MCT is an MCS that is closed under 
applications of rule (INT). Hence, every MCT is an MCS, and every MCS with this additional closure property is an MCT. 
Intuitively, closure under rule (INT) helps with the construction of a canonical model for formulae like the ones in
Ex.~\ref{ex:isectproblem} by introducing a new atomic proposition which can be used to distinguished two copies of worlds that
would otherwise be equal as MCSs, but are not equal as MCTs.

Balbiani and Vakarelov then claim that every consistent formula is satisfiable in \emph{the} canonical \pdocap model (based
on MCTs), c.f.\ \cite[Prop.~6.3]{DBLP:journals/fuin/BalbianiV01}. This is not true when taken literally, instead, their constructions
prove the following weaker statement: every consistent \pdocap formula is satisfiable in \emph{some} canonical model. This is a
simple consequence of the fact that applications of rule (INT) introduce new atomic propositions that were not present in the
language in the first place. In other words: the canonical model $\mathcal{K}$ whose worlds are MCTs depends on the 
underlying language.

\begin{example}
\label{ex:isectproblemfull}
Consider \pdocap over the empty set of atomic propositions and let $\Phi_0$ be the theory of the world with no successors. Clearly,
$\Phi_0$ is satisfiable and can therefore not be inconsistent with respect to a sound axiomatisation. As argued above, 
$\mathit{Split}(\Phi_0)$ is also satisfiable but its models must contain two disjoint copies of models of $\Phi_0$. Given that 
$\Phi_0$ is maximal, i.e.\ an MCS, and that rule (INT) is not applicable when no propositions are available and therefore every MCS is
already an MCT we get that $\mathit{Split}(\Phi_0)$ is \emph{not} satisfiable in \emph{the} canonical for \pdocap over the empty
set of atomic propositions. This shows that Prop.~6.3 of \cite{DBLP:journals/fuin/BalbianiV01} needs to be taken with care, namely
in the weaker sense stated above. Note that $\mathit{Split}(\Phi_0)$ is logically equivalent to a finite set of \pdocap formulae, and, 
hence, to a single formula.
\end{example}
Example~\ref{ex:isectproblemfull} can be extended to any given set of propositions. Let $\Phi$ be any propositional labelling of the world with no
successors that is complete in the sense that for every proposition $p$ of the underlying $\mathcal{P}$ we have $p \in \Phi$ iff $\neg p \not\in \Phi$.
Then $\mathit{Split}(\Phi)$ requires two copies of this world to exist. Any invocation of the rule (INT) would require an \emph{additional}
proposition to distinguish the two. Consequently, the calculus of \cite{DBLP:journals/fuin/BalbianiV01} considers the logic \pdocap
in the language of some vocabulary $\tau$ but makes use of formulae that belong to the language of \pdocap in the vocabulary of a genuine 
superset $\tau^*$ of $\tau$. It does not help to consider the larger vocabulary $\tau^*$ in the first place as the characterisation of
the intersection operator uses propositions for every subset of a model, c.f.\ \cite[Prop.~2.1]{DBLP:journals/fuin/BalbianiV01}. This
is clearly problematic in a canonical model construction when the propositions that are used to form the worlds are derived from the
set of all subsets of worlds.  

Balbiani has refined the construction of \cite{DBLP:journals/fuin/BalbianiV01} in order to get rid of the need to introduce new propositions
\cite{Ba2003.6}. Weak completeness is proved using a similar canonical model construction, strong completeness is not achieved. 

Finally, Balbiani and Vakarelov have also extended their work on \pdocap to the full \pdtcap \cite{DBLP:journals/jancl/BalbianiV03} also
using a very similar canonical model construction based on a rule which requires new atomic propositions, and therefore statements about
these models need to be taken with similar care. Most importantly, there is no unique canonical model for all of \pdtcap because its structure 
depends on the vocabulary of the underlying logic but the correctness proofs require its MCTs to be built using propositions for every 
subset of the model. At last, possible problems with canonical models for non-compact logics are avoided using an infinitary rule to handle
Kleene stars in programs.

%In other words, let $P$ be the set of all atomic propositions. Note that $P$ can be supposed to be a set rather than a class because
%atomic propositions are syntactic objects. Let $\Phi_{\infty}$ be the theory of the world with no successors and an arbitrary
%labelling with propositions from $P$. Again, $\Phi_{\infty}$ is satisfiable and maximal and therefore an MCS with respect to any
%axiomatisation that is sound. It is also an MCT because closure under rule (INT) cannot introduce any further atomic propositions 
%that are not already present in $P$. Again, $\mathit{Split}(\Phi_\infty)$ is satisfiable but it is not satisfiable in any canonical
%model for \pdocap over an exhaustive vocabulary.
%\end{example}

The following sections are devoted to the presentation of a sound and complete axiomatisation for \pdocap over an arbitrary 
vocabulary that works in the very same vocabulary.

%%% Local Variables:
%%% mode: latex
%%% TeX-master: "pdl"
%%% End:

\section{Axiomatising \pdocap}
In this section we propose an axiomatisation for \pdocap and derive a normal form such that every \pdocap formula
is equivalent to one in normal form, and this equivalence is also provable in the calculus. 
 
\subsection{Axioms and Rules} 
% !TEX root =  pdl.tex

\begin{figure}[t]
    \begin{center}
      \begin{minipage}{6cm}
	\begin{flalign}
		&\Mubox{\alpha}p \biimp \neg \Mudiam{\alpha}\neg p \tag{\texttt{Dl}}\label{ax:duality} \\
		&\Mudiam{p?}q \biimp p \wedge q \tag{\texttt{?}}\label{ax:?}\\
		&\Mudiam{\alpha \cap p?}q \biimp \Mudiam{\alpha^\circlearrowleft}(p\wedge q) \tag{\texttt{T1}}\label{ax:t1}\\
		&\Mubox{\alpha;\beta}p \biimp \Mubox{\alpha}\Mubox{\beta}p \tag{\texttt{;}}\label{ax:semicolon} \\
		&\Mudiam{\alpha \cup \beta}p \biimp \Mudiam{\alpha}p \vee \Mudiam{\beta}p \tag{\texttt{D}}\label{ax:cup}\\[2mm]
			&\alpha \cap \beta \Rightarrow \alpha \tag{\texttt{Wk}}\label{ax:weak}\\
			&\alpha \cap \beta \Leftrightarrow \beta \cap \alpha \tag{\texttt{Cm}}\label{ax:comm}\\
			&\alpha \cap \alpha \Leftrightarrow \alpha \tag{\texttt{Ct}}\label{ax:contraction} \\
	        &(\alpha \cup \beta) ; \gamma \Leftrightarrow (\alpha ; \gamma)  \cup (\beta ; \gamma) \tag{\texttt{D3}}\label{ax:distr3}\\
			&\alpha ; (\beta \cup \gamma) \Leftrightarrow (\alpha ; \beta)  \cup (\alpha ; \gamma) \tag{\texttt{D4}}\label{ax:distr4}\\
		&\alpha \cap p? \Leftrightarrow (\Mudiam{\alpha^\circlearrowleft}p)? \tag{\texttt{T}}\label{ax:cyctotest}
	\end{flalign}
	\end{minipage}\hfill
	\begin{minipage}{8cm}
	\begin{flalign}
		&\Mubox{\alpha}(p \imp q) \imp \Mubox{\alpha}p \imp \Mubox{\alpha}q \tag{\texttt{K}}\label{ax:k} \\
		&\Mudiam{\alpha^\circlearrowleft}p \wedge \Mudiam{\beta^\circlearrowleft}q \imp \Mudiam{(\alpha;\beta)^\circlearrowleft}(p\wedge q) 
		     \tag{\texttt{C1}}\label{ax:c1}\\
		&\Mubox{\alpha^\circlearrowleft}p \wedge \Mubox{\beta^\circlearrowleft}p \imp \Mubox{\alpha^\circlearrowleft;\beta^\circlearrowleft}p
		     \tag{\texttt{C2}}\label{ax:c2} \\
         &\Mudiam{\alpha^\circlearrowleft}p \wedge \Mubox{\alpha^\circlearrowleft}q \imp p \wedge q \tag{\texttt{C3}}\label{ax:c3} \\
        &\Mudiam{\alpha;(p\vee q)?;\beta}r \biimp \Mudiam{\alpha;p?;\beta}r \vee \Mudiam{\alpha;q?;\beta}r \tag{\texttt{V}}\label{ax:V}\\[2mm]
			&\alpha \cap (\beta \cap \gamma) \Leftrightarrow (\alpha \cap \beta) \cap \gamma \tag{\texttt{A}}\label{ax:assoc}\\
			&(\alpha;p?) \cap \beta \Leftrightarrow (\alpha \cap \beta);p? \tag{\texttt{T2}}\label{ax:t2}\\
			&(p?;\alpha)\cap \beta) \Leftrightarrow p?;(\alpha\cap\beta) \tag{\texttt{T3}}\label{ax:t3} \\
				&\alpha \cap (\beta \cup \gamma) \Leftrightarrow (\alpha \cap \beta) \cup (\alpha \cap \gamma) \tag{\texttt{D1}}\label{ax:distr1}\\
				&\alpha \cup (\beta \cap \gamma) \Leftrightarrow (\alpha \cup \beta) \cap (\alpha \cup \gamma) \tag{\texttt{D2}}\label{ax:distr2}\\
	&(\varphi \leftrightarrow \psi) \leftrightarrow (\varphi? \Leftrightarrow \psi?) \tag{\texttt{TP}}\label{ax:testprog}
	\end{flalign}
    \end{minipage} \\
%     \begin{minipage}{4.5cm}
%       \begin{flalign}
%          &\alpha \cap p? \biimp (\Mudiam{\alpha^\circlearrowleft}p)? \tag{\texttt{T}}\label{ax:cyctotest}
%       \end{flalign}
%     \end{minipage} \hfill    
    \begin{minipage}{10.5cm}
      \begin{flalign}
	&\alpha^\circlearrowleft \Rightarrow \alpha \ullcorner (\beta_2;\beta_3;[(\beta_1;\beta_2;\beta_3)^\circlearrowleft]p? /
        \beta_2;[(\beta_3;\beta_1;\beta_2)^\circlearrowleft]p?;\beta_3 \ulrcorner^\circlearrowleft \nonumber \tag{\texttt{C}}\label{ax:apo}
      \end{flalign}
    \end{minipage}
    
%     \begin{minipage}{10.5cm}
%       \begin{flalign}
%     (\varphi \leftrightarrow \psi) \leftrightarrow (\varphi? \leftrightarrow \psi?) \tag{\texttt{TP}}\label{ax:testprog}
%     \end{flalign}
%     \end{minipage}
  \end{center}

%\vspace*{-6mm}
% 	\begin{center}
% 	    \begin{minipage}{0.45\textwidth}
% 		\begin{align}
% 		    \alpha \cap p? &\biimp (\Mudiam{\alpha^\circlearrowleft}p)? \tag{\texttt{T}}\label{ax:cyctotest}
% 		\end{align}
% 	    \end{minipage}
% 	\end{center}		
% \vspace*{-6mm}
% 			\begin{align}
% 					\alpha^\circlearrowleft \imp \alpha\big[(\beta_2;\beta_3;[(\beta_1;\beta_2;\beta_3)^\circlearrowleft]p? / \beta_2;[(\beta_3;\beta_1;\beta_2)^\circlearrowleft]p?;\beta_3 \big]^\circlearrowleft \nonumber \tag{\texttt{C}}\label{ax:apo}
% 			\end{align}
% \vspace*{-6mm}

\caption{The formula and program axioms for \pdocap.}
\label{fig:axioms}
\end{figure}

Let $\Delta$ be the smallest proof calculus that contains all propositional tautologies, the formula axiom schemes and
program axioms schemes shown in Fig.~\ref{fig:axioms} and the inference rules 
\begin{mathpar}
(\texttt{MP})\ \inferrule{\varphi \\ \varphi \imp \psi}{\psi} \and 
(\texttt{Gen})\ \inferrule{\varphi}{[\alpha]\varphi} \and 
(\texttt{USub})\ \inferrule{\varphi}{\varphi \ullcorner \psi/p \ulrcorner} \and
(\texttt{PSub})\ \inferrule{\varphi \\ \alpha \Rightarrow \alpha'}{\varphi\ullcorner\Mudiam{\alpha'} / \Mudiam{\alpha}\ulrcorner}
\end{mathpar}
where $\varphi\ullcorner\psi/p\ulrcorner$ is the usual substitution and $\varphi\ullcorner\Mudiam{\alpha'} / \Mudiam{\alpha}\ulrcorner$ is meant to denote that every program $\Mudiam{\alpha}$ which occurs under an even number of negation symbols in the syntax tree of the formula is being replaced by $\Mudiam{ \alpha' }$. We use $\ullcorner \hspace{1ex} \ulrcorner$ instead of the usual brackets for the substitution operator to distinguish them from the box modality.

We write $\vdash \varphi$ if the \pdocap-formula $\varphi$ can be derived from the axioms of $\Delta$ alone by repeated application of the rules of inference. For a set $\Phi$ of \pdocap-formulae we write $\Phi \vdash \varphi$ if there exist $\varphi_1,\dotsc,\varphi_n \in \Phi$ such that $\vdash (\bigwedge_{i=1}^n \varphi_i) \imp \varphi$. 

The purpose of rule~\eqref{ax:apo} is to deal with properties of cyclic structures. 

\begin{example}
Consider the formula 
$\varphi =\Mudiam{(a;\psi?;b)^\circlearrowleft}\true$ which is satisfied at a state $u$ which is the beginning of an $(a;b)$-cycle such that $\psi$ is satisfied at an intermediate state $v$ with $\Transition{u}{a}{v}$ and $\Transition{v}{b}{u}$. The satisfiability of $\varphi$ depends on $\psi$, not just on whether it is satisfiable itself but also whether it is compatible with being satisfied on a cyclic structure.

Consider $\psi = \Mubox{(b;a)^\circlearrowleft}\false$. Clearly, $\psi$ can not be satisfied on a state $v$ with $\Transition{v}{b}{u}$ and $\Transition{u}{a}{v}$ for some $u$ whence $\varphi$ is not satisfiable at all. Rule~\ref{ax:apo} incorporates this into the calculus: Considering $\varphi$ with $\psi = \Mubox{(b;a)^\circlearrowleft}\false$, a combination of rule~\ref{ax:apo} and modus ponens yields that $\Mudiam{(a;b;\psi'?)^\circlearrowleft}\true$ with $\psi' = \Mubox{(a;b)^\circlearrowleft}\false$ is a logical consequence of $\varphi$. Using axiom~\ref{ax:t2} we also obtain $\Mudiam{(a;b)^\circlearrowleft}\psi'$ with $\psi'$ as before as a logical consequence. Using axiom~\ref{ax:c3} with $q = \true$ we obtain that $\psi'$ is a also a logical conseqence of $\varphi$ which makes $\varphi$ inconsistent after a few derivations, correctly reflecting unsatisfiability of $\varphi$.

The intuition behind rule~\ref{ax:apo} is that it allows tests on cyclic programs to be transferred further along the cycle while correctly adjusting programs in these tests for the fact that they have been transferred \emph{and} accounting for the fact that all this occurs on a cycle. 
\end{example}

\begin{lemma}
\label{lem:vdashsound}
$\vdash$ is sound, i.e., $\Phi \vdash \varphi$ only if $\Phi \models \varphi$.
\end{lemma}
The proof is by standard induction on the length of a proof. The rest of the paper is devoted to showing completeness of $\vdash$ using
the notion of an \mcs, c.f.\ Sect.~\ref{sec:canonical}. 
By Zorn's Lemma, every consistent formula set is contained in an \mcs. Moreover, if $\Phi$ is an \mcs, then it has the following properties:
(1) $\Phi$ is closed under $\vdash$; (2) $\varphi \wedge \psi \in \Phi$ iff $\varphi \in \Phi$ and $\psi\in \Phi$; 
(3) $\varphi \vee \psi \in \Phi$ iff $\varphi \in \Phi$ or $\psi\in \Phi$; and (4) $\varphi \in \Phi$ iff $\neg \varphi \notin \Phi$
for any formula $\varphi$.

\begin{lemma}
\label{lem:intermediate-set-cons}
Let $\Phi,\Psi$ be \mcs, $X$ be a consistent set that is closed under $\vdash$ of \pdocap-formulae and $\alpha_1$ and $\alpha_2$ be \pdocap-programs, $\chi \in X$ such that 
$\Mudiam{\alpha_1;\chi?;\alpha_2}\psi \in \Phi$ for all $\psi \in \Psi$. Then the set 
$X' \coloneqq X \cup \{\varphi \mid \Mubox{\alpha_1}\varphi\in\Phi\} \cup \{\Mudiam{\alpha_2}\psi \mid \psi \in \Psi\}$ is consistent.
\end{lemma}
\begin{proof}
Assume $X'$ was not consistent. Then, w.l.o.g., there are $\varphi,\psi,\chi$ such that $\chi \in X$, 
$\Mubox{\alpha_1}\varphi \in \Phi$ and $\psi\in \Psi$, but 
$\vdash \varphi \wedge \chi \wedge \Mudiam{\alpha_2}\psi\imp \false$. We have $\Mudiam{\alpha_1;\chi?;\alpha_2}\psi \in \Phi$ and, 
hence, $\Mudiam{\alpha_1}(\chi \wedge \Mudiam{\alpha_2}\psi) \in\Phi$ as well as $\Mubox{\alpha_1}\varphi \in \Phi$, so 
$\Mudiam{\alpha_1}(\varphi \wedge \chi \wedge \Mudiam{\alpha_2}\psi) \in \Phi$ which contradicts the assumption.
\end{proof}

\begin{lemma}
\label{lem:loop-set-cons}
Let $\Phi,\Psi,\Upsilon$ be \mcs, $X$ be consistent and closed under $\vdash$, and $\alpha_1, \alpha_2,\beta_1,\beta_2$ be programs. Let
\begin{enumerate} 
\item $\Mudiam{\alpha_1;\chi?;\alpha_2}\psi \in \Phi$ for all $\psi \in \Psi$, $\chi \in X$,
\item $\Mudiam{\big((\beta_2;\upsilon?;\beta_1);\varphi?;(\alpha_1;\chi?;\alpha_2)\big)^\circlearrowleft}\true \in \Psi$ 
      for all $\upsilon \in \Upsilon, \varphi\in \Phi,\chi \in X$,
\item $\Mudiam{\big((\alpha_1;\chi?;\alpha_2);\psi?;(\beta_2;\upsilon?;\beta_1)\big)^\circlearrowleft}\true \in \Phi$ 
      for all $\chi \in X, \psi \in \Psi, \upsilon \in \Upsilon$,
\item $\Mudiam{((\beta_1;\varphi?;\alpha_1;);\chi?;(\alpha_2;\psi?;\beta_2))^\circlearrowleft}\true \in \Upsilon$ for all 
      $\varphi \in \Phi,\chi \in X, \psi \in \Psi$.
\end{enumerate}
Then the following set $X^*$ is consistent. 
\begin{align*}
X^* \enspace = \enspace X \ \cup \ & \{\varphi\mid \Mubox{\alpha_1}\varphi\in\Phi\}  \cup \{\Mudiam{\alpha_2}\psi \mid \psi \in \Psi\} \\
  \ \cup \ & \{\Mudiam{\big((\alpha_2;\psi?;\beta_2);\upsilon?;(\beta_1;\varphi?;\alpha_1)\big)^\circlearrowleft}\true \mid \psi \in \Psi,\upsilon\in \Upsilon,\varphi\in \Phi\}
\end{align*}
\end{lemma}
\begin{proof}
We show that the union of $X$ and the third set is consistent. 
Assume that it is not. Then there are finitely many formulae $\zeta_1,\dotsc,\zeta_n$ of the form
$\zeta_i = \Mudiam{((\alpha_2;\psi?;\beta_2);\upsilon?;(\beta_1;\varphi?;\alpha_1))^\circlearrowleft}\true$,
with programs and formulae as suggested such that $\zeta = \bigvee_{i = 1}^n \neg \zeta_i \in X$. Then 
$\Mudiam{((\beta_1;\varphi?;\alpha_1;);\zeta?;(\alpha_2;\psi?;\beta_2))^\circlearrowleft}\true \in \Upsilon$
and, since $\Upsilon$ is an \mcs, also rule \eqref{ax:V} implies that also
$\Mudiam{((\beta_1;\varphi?;\alpha_1;);\neg\zeta_i?;(\alpha_2;\psi?;\beta_2))^\circlearrowleft}\true \in \Upsilon$
for at least one $i$. Now we can apply rule \eqref{ax:apo} to conclude that 
$\Mudiam{((\beta_1;\varphi?;\alpha_1;);\true?;(\alpha_2;\psi?;\beta_2);\neg\zeta'_i))^\circlearrowleft}\true \in \Upsilon$
where $\neg \zeta'_i$ is equivalent to 
$\Mubox{((\beta_1;\varphi?;\alpha_1);\true?;(\alpha_2;\psi?;\beta_2);\upsilon?)^\circlearrowleft}\false$
which contradicts consistency of $\Upsilon$ since then also $\neg \zeta'_i \in \Upsilon$.

Similar arguments show that we can also add the other two sets without losing consistency.
\end{proof}

%%% Local Variables:
%%% mode: latex
%%% TeX-master: "pdl"
%%% End:

\subsection{A Normal-Form Lemma for \pdocap}
% !TEX root =  pdl.tex

We partition the \pdocap-programs into two groups. The first one consists of \emph{cyclic} programs, \Cyc in short.
They test if something holds at the present state, possibly requiring this state to have a (perhaps complex) self-loop. Thus, 
they force the evaluation of a formula to stay at the current node. Secondly, there are the \Forw-programs which make up all
others. These require the evaluation of a formula to take at least one step into some direction. Syntactically, cyclic and
forward programs are defined as follows:
\begin{displaymath}
\alpha_{\forw} \enspace \coloneqq \enspace a \mid \alpha_{\forw} \cap \alpha_{\forw} \mid \alpha_{\forw} ; \alpha_{\cyc} ; \alpha_{\forw}
 \qquad \qquad
\alpha_{\cyc} \enspace \coloneqq \enspace \varphi? \mid \alpha_{\forw} \cap \varphi?
\end{displaymath}

\begin{lemma}
\label{lem:program-normal-form}
For every $\varphi$ there is a $\varphi'$ such that $\vdash \varphi \biimp \varphi'$ and all programs in $\varphi'$ belong to
$\Cyc \cup \Forw$. Moreover, for all formulae of the form $\Mudiam{\alpha_{\cyc}}\varphi$ we have 
$\vdash \Mudiam{\alpha_{\cyc}}\varphi \biimp \Mudiam{\alpha_{\forw}^\circlearrowleft}\varphi'$ for some formula $\varphi'$ or 
just $\vdash \Mudiam{\alpha_{\cyc}}\varphi \leftrightarrow \varphi'$ for some formula $\varphi'$ without the modal prefix 
$\Mudiam{\alpha_{\cyc}}$.
\end{lemma}
\begin{proof}
	We will show this in several steps. First, we eliminate the disjunction operator from programs. Using \eqref{ax:comm},
	\eqref{ax:distr1},\ldots,\eqref{ax:distr4}, we can transform every program into one in which $\cup$ does not occur underneath a
	different program operator. Using axiom \eqref{ax:cup} it is possible to eliminate occurrences of $\cup$ in such top-level positions in programs. In the following, $\alpha_{\cyc}$ denotes an arbitrary program from \Cyc.
	
	Next we note that \Cyc-programs commute with the intersection operator if they are at the end or the beginning of a sequential composition: Using axiom \eqref{ax:cyctotest} and some basic propositional logic, one can check that $\vdash \alpha \cap \psi? \enspace \Leftrightarrow \enspace (\Mudiam{\alpha^\circlearrowleft}\true \wedge \psi)?$ and thus every \Cyc-program is equivalent to a test. Consequently with \eqref{ax:t2}, \eqref{ax:t3} we can derive that	
	\[
			\vdash (\alpha_{\cyc} ; \alpha) \cap \alpha' \Leftrightarrow \alpha_{\cyc} ; ( \alpha \cap \alpha' ) \text{ and }
			\vdash (\alpha ; \alpha_{\cyc}) \cap \alpha' \Leftrightarrow ( \alpha \cap \alpha' ) ; \alpha_{\cyc}.
	\]
	Note that after using this equivalence $\alpha_{\cyc}$ is at the beginning or the end of the respective sequence. We can thus assume that every occurence of a \Cyc-program is as high as possible in the syntax-tree of a program.
	
	In the next step we want to eliminate, resp.\ simplify isolated \Cyc-programs, i.e.\ programs $\alpha_{\cyc}$ in a 
	formula $\Mudiam{\alpha_{\cyc}}\varphi$. Using \eqref{ax:t1} we get 
	\[
	\vdash \Mudiam{\psi?}\varphi \Leftrightarrow \varphi \wedge \psi \qquad \text{and} \qquad  
	\vdash \Mudiam{\psi? \cap \alpha} \varphi \Leftrightarrow \Mudiam{\alpha \cap \true?}(\varphi \wedge \psi)\ .
    \]
    Note that these equivalences entail the second and third statement of the lemma.
    	 
	Similarly we can simplify \Cyc-programs at the beginning or end of a program that is just a sequential composition by first using \eqref{ax:semicolon} and then \eqref{ax:t1} again.
%	\[
%		\vdash \Mudiam{\psi? ; \alpha}\varphi \biimp \psi \wedge \Mudiam{\alpha}\varphi \qquad \text{and} \qquad
%		\vdash \Mudiam{\alpha ; \psi?}\varphi \biimp \Mudiam{\alpha}(\varphi \wedge \psi)
%	\]

	We have just dealt with \Cyc-programs at the beginning or the end of a sequence of sequential compositions and moved them up in the syntax tree wherever possible. It remains to be seen how they are treated in the middle of such a sequence.
	
	Let $\alpha ; \alpha_{\cyc} ; \alpha'$ be a program with a cyclic program in the middle of a sequence. We make a case distinction over $\alpha$. The part for $\alpha'$ is symmetric.	
	In case $\alpha$ is atomic or the intersection of \Forw-programs, we are finished. 
	%If $\alpha$ is a test, or another $\Cyc$-program, 	we separate this program as before using \eqref{ax:semicolon} and then continue to separate $\alpha_{\cyc}$ as well.
	If $\alpha = \beta ; \beta_{\cyc}$, then $\alpha ; \alpha_{\cyc} ; \alpha' = \beta ; \beta_{\cyc} ; \alpha_{\cyc} ; \alpha '$. Again, we use the fact that each \Cyc-program is equivalent to a simple test. Using \eqref{ax:semicolon} and \eqref{ax:?} several times we get $\vdash \beta ; \beta_{\cyc} ; \alpha_{\cyc} ; \alpha' \Leftrightarrow \beta ; \beta_{\cyc} \cap \alpha_{\cyc} ; \alpha$.
	Note, that \Cyc-programs are closed under intersections. 
	
	%The case of $\alpha = \beta_{\cyc} ; \beta$ can again be handled by separating the cyclic program first (or moving it up the syntax-tree of the program by the equivalences above).	
	The last case is that of $\alpha = \beta ; \beta'$ where both are simple \Forw-programs. Using basic propositional logic,  
	\eqref{ax:semicolon} and \eqref{ax:?} we get $\vdash \beta ; \beta' \Leftrightarrow \beta ; \true? ; \beta'$.

	%Finally, we separate \Cyc-programs at the beginning or the end of a sequence of sequential compositions that may arise by inductively applying the equivalences above.
\end{proof}

%\begin{lemma}
%	\label{lem:cycequalstest}
%	For each \cyc-program $\alpha_{\cyc}$ there is a simple test-program $\varphi?$ such that $\vdash \alpha_{\cyc} \leftrightarrow \varphi?$. Furthermore, in $\varphi$ there are only tests against $\true$.
%\end{lemma}
%\begin{proof}
%	Without loss of generality, let $\beta_{\cyc} = \alpha_{\forw} \cap \psi?$. Thus, $\beta_{\cyc}$ indicates a loop and requires $\psi$ to hold at the current state. Using axiom \eqref{ax:cyctotest} and some basic propositional logic, one can check that $\vdash \alpha_{\forw} \cap \psi? \enspace \leftrightarrow \enspace (\Mudiam{\alpha_{\forw}^\circlearrowleft}\true \wedge \psi)?$. 
%	Using this equivalence recursively to rewrite $\alpha_{\cyc}$ yields the desired result.
%\end{proof}

As a consequence of the proof of this lemma, we will sometimes assume w.l.o.g.\ that \cyc-programs are of the form $\varphi?$ for some suitable $\varphi$.

\section{A Canonical Model for \pdocap}
\subsection{Large Programs}
In order to maintain induction invariants during the construction of our canonical model, we need to allow tests to test against arbitrary sets of formulae as opposed to just one formula. We call the resulting extension of programs \emph{large programs}. This is not exactly the same notion of large programs as in \cite{DBLP:journals/jancl/BalbianiV03}.% \todo{Ggf. umbenennen.}
\begin{definition}
The set of \emph{large programs} is defined inductively via
\[
\alpha \Coloneqq a \mid \alpha \cap \alpha \mid \alpha;\Phi?,\alpha
\]
where $\Phi$ is a consistent set of \pdocap-formulae. A large loop is of the form $\alpha^\circlearrowleft$, where $\alpha$ is a large program.

An ordinary program $\alpha$ is an \emph{instance} of a large program $\alpha_l$ if 
\begin{itemize}
\item $\alpha = \alpha_l = a$ for some accessibility relation $a$ or,
\item $\alpha = \alpha^1 \cap \alpha^2, \alpha^l = \alpha^1_l \cap \alpha^2_l$ and $\alpha^i$ is an instance of $\alpha^i_l$ for $i = 1,2$ or,
\item $\alpha = \alpha^1;\varphi?;\alpha^2, \alpha_l = \alpha^1_l;\Phi?;\alpha^2_l$, $\varphi \in \Phi$ and $\alpha^i$ is an instance of $\alpha^i_l$ for $i = 1,2$.
\end{itemize}
A loop $\alpha^\circlearrowleft$ is an instance of a large loop $\alpha_l^\circlearrowleft$ if $\alpha$ is an instance of $\alpha_l$.

We write $\alpha \leq \alpha'$ for large programs $\alpha, \alpha'$ if and only if every instance of $\alpha$ is an instance of $\alpha'$.
\end{definition}
A large program is an ordinary program where tests against a formula have been replaced with tests against a set of formulae. Clearly, an ordinary program with consistent tests can be made large by replacing tests of the form $\varphi?$ by tests of the form $\{\varphi\}?$, and the original program will be an instance of the new large program.
\begin{definition}
Let $\Phi, \Psi$ be \mcs and let $\alpha$ be a large program. % such that for no instance $\alpha'$ of $\alpha$ and no $\psi \in \Psi$ we have that $[\alpha]\neg \psi \in \Phi$. 
For occurrences of subprograms $\beta$ define the left and right sets $l(\beta)$ and $r(\beta)$ in a top-down manner via
\begin{itemize}
\item $l(\alpha) \coloneqq \Phi, r(\alpha) \coloneqq \Psi$,
\item If $\alpha = \alpha_1 \cap \alpha_2$ then $l(\alpha_1) = l(\alpha_2) \coloneqq l(\alpha)$ and $r(\alpha_1) = r(\alpha_2) \coloneqq r(\alpha)$,
\item If $\alpha = \alpha_1; X?;\alpha_2$ then $l(\alpha_1) \coloneqq l(\alpha), r(\alpha_1) = l(\alpha_2) \coloneqq X, r(\alpha_2) \coloneqq r(\alpha)$.
\end{itemize}
In the case of a large loop $\alpha^\circlearrowleft$ and a set $\Phi$, % such that for no instance $\alpha'$ of $\alpha$ and no $\varphi \in \Phi$ we have $[\alpha'^\circlearrowleft]\neg \varphi$, 
define the left and right programs $\lp(X)$ and $\rp(X)$ of an occurrence of a test $X$ in a top-down manner as follows:
\begin{itemize}
\item $\lp(\Phi) \coloneqq \true?, \rp(\Phi) \coloneqq \true?$,
%\item If $\alpha = \alpha_1 \cap \alpha_2$ then $lp() = lp(\alpha_2) = lp(\alpha)$ and $rp(\alpha_1) = rp(\alpha_2) = rp(\alpha)$,
\item If $\alpha = \alpha_1; X?;\alpha_2$ then $\lp(X) \coloneqq \lp(l(\alpha));l(\alpha)?;\alpha_1$ and $\rp(X) \coloneqq
  \alpha_2;r(\alpha)?,\rp(r(\alpha))$.
\end{itemize}
We say that $\Transition{\Phi}{\alpha}{\Psi}$ is \emph{consistent} if $\Mubox{\beta'}\neg \psi \notin l(\beta)$ for all instances $\beta'$ of subprograms $\beta$ of $\alpha$ and all $\psi \in r(\beta)$, and all test sets are consistent. We say that $\Transition{\Phi}{\alpha}{\Psi}$ is inconsistent if it is not consistent. A large loop $\alpha^\circlearrowleft$ is consistent at $\Phi$ if the above holds with $\Phi = \Psi$ and for each test set $X$, no formula of the form $\Mubox{(\beta_1;\varphi?;\beta_2)^\circlearrowleft}\false$, with $\beta_1$ an instance of $\rp(X)$, $\beta_2$ an instance of $\lp(X)$ and $\varphi \in \Phi$ is in $X$.

We say that $\Transition{\Phi}{\alpha}{\Psi}$ is \emph{maximally consistent} if $\Transition{\Phi}{\alpha}{\Psi}$ is consistent and every test in $\alpha$ is an \mcs. In particular, for all subprograms $\beta = \beta_1;X?;\beta_2$ we have that $X \supseteq \{\varphi\mid\Mubox{\beta'_1}\varphi \in l(\beta), \beta'_1\text{ instance of } \beta_1\} \cup \{\Mudiam{\beta'_2} \psi \mid \psi \in r(\beta), \beta'_2 \text{ instance of } \beta_2 \}$.
We say that a large loop $\alpha^\circlearrowleft$ is maximally consistent at $\Phi$ if $\Transition{\Phi}{\alpha}{\Phi}$ is maximally consistent and, additionally, for all tests $X?$, 
\[
\{\Mudiam{(\beta_1;\varphi?\beta_2)^\circlearrowleft}\true \mid \beta_1\text{ instance of } \lp(X), \beta_2\text{ instance of } \rp(X),\varphi\in \Phi\} \subseteq X.
\]
\end{definition}
\begin{lemma}
\label{lem:large-prog-max}
Let $\Transition{\Phi}{\alpha}{\Psi}$ for a large program $\alpha$ be consistent. Then there is a large program $\alpha' \geq \alpha$ such that $\Transition{\Phi}{\alpha'}{\Psi}$ is maximally consistent. Moreover, if $\alpha^\circlearrowleft$ is a large loop such that $\Transition{\Phi}{\alpha^\circlearrowleft}{\Phi}$ is consistent, there is $\alpha' \geq \alpha$ such that $\Transition{\Phi}{\alpha'}{\Phi}$ is maximally consistent.
\end{lemma}
\begin{proof}
Let $\Psi, \Phi$ and $\alpha$ be as in the lemma. For the case of a large loop, set $\Psi =\Phi$. For convenience, we assume that the test sets in every subprogram of $\alpha$ are closed under conjunctions, if not, we close them under conjunctions. Clearly this will not make $\Transition{\Phi}{\alpha}{\Psi}$ inconsistent. 

The proof proceeds recursively: Assume that $\beta$ is a subprogram of $\alpha$ such that $l(\beta)$ and $r(\beta)$ are already \mcs. There are three cases: $\beta = a$ for some atomic program $a$, $\beta = \beta_1 \cap \beta_2$ and $\beta = \beta_1;X?,\beta_2$. In the first case, there is nothing left to do. In the second case, convert all the test sets in $\beta_1$ into \mcs such that $\Transition{\Phi}{\alpha}{\Psi}$ stays consistent. Then repeat the same procedure with $\beta_2$.

In the third case, we have to find an \mcs $X^*$ that is a superset of $X$ such that replacing $X$ by $X^*$ will not break consistency of $\Transition{\Phi}{\alpha}{\Psi}$. W.l.o.g., $X$ is already closed under $\vdash$. By Lemma~\ref{lem:intermediate-set-cons}, the set $X' = X \cup \{\varphi\mid\Mubox{\beta_1}\varphi \in l(\beta)\}\cup \{\Mudiam{\beta_2}\psi\mid\psi \in r(\beta)\}$ is consistent. Moreover, replacing $X$ by $X'$ will not make $\Transition{\Phi}{\alpha}{\Psi}$ inconsistent because otherwise, $\Phi$ would be inconsistent. By Lemma~\ref{lem:loop-set-cons}, in the case of a large loop, set $X'$ as 
\begin{align*}
X \cup &  \{\varphi\mid\Mubox{\beta_1}\varphi \in l(\beta)\} \\
        \cup & \{\Mudiam{\beta_2}\psi\mid\psi \in r(\beta)\} \\
        \cup & \{\Mudiam{((\beta_2;\psi?;\rp(r(\beta)));\upsilon?;(\lp(l(\beta));\varphi?;\beta_1))^\circlearrowleft}\true\mid\psi\in r(\beta),\upsilon \in \Phi, \varphi\in l(\beta) \} %\\
      %  \cup & \{\Mudiam{(rp(r());\varphi?;(lp(\beta);\psi?;\beta)\cap\true?)}\true \mid \varphi\in\Phi,\psi \in l(\beta)\}
\end{align*} 
which is also consistent.
Now consider the set $\mathcal{X}$ of all consistent supersets of $X'$ such that \Transition{\Phi}{\alpha}{\Psi} stays consistent if $X$ is replaced by a set from $\mathcal{X}$. This set is nonempty, because it contains $X'$. It is also partially ordered by set inclusion and the union of any chain of sets from $\mathcal{X}$ is in $\mathcal{X}$ for otherwise, there would be a minimal set in the chain which also is not in $\mathcal{X}$. 

By Zorn's Lemma, $\mathcal{X}$ contains a maximal element $X^*$. We argue that $X^*$ is an \mcs. Assume otherwise, then there is $\varphi \in \pdocap$ such that neither $\varphi$ nor $\neg \varphi$ are in $X^*$. Since $X^*$ is consistent, one of $\varphi$ and $\neg \varphi$ can be added to $X^*$ without losing consistency of $X^*$. Hence, by assumption, replacing $X$ in $\beta$ by either of $X \cup \{\varphi\}$ and $X \cup \{\neg \varphi\}$ will make $\Transition{\Phi}{\alpha}{\Psi}$ inconsistent, but using $X^*$ itself does not. Then there are instances $\alpha_1$ and $\alpha_2$ of $\alpha$ such that both $\Mubox{\alpha_1\ullcorner(\beta_1;\varphi?;\beta_2)/\beta\ulrcorner} \neg \psi_1$ and $\Mubox{\alpha_2\ullcorner\beta_1;\neg\varphi?;\beta_2/\beta\ulrcorner} \neg \psi_2$ are in $\Phi$ for $\psi_1,\psi_2\in \Psi$. Since $\Psi$ and all test sets are closed under conjunction, we can assume that $\alpha_1 = \alpha_2$ and $\psi_1 = \psi_2$. But $\vdash \Mubox{\alpha_1\ullcorner(\beta_1;(\varphi \vee \neg \varphi)?;\beta_2)/\beta\ulrcorner}\neg \psi_1 \leftrightarrow \Mubox{\alpha_1}\neg\psi$ and $\Mubox{\alpha_1}\neg\psi \in \Phi$, which is a contradiction to $X^*$ being a safe replacement for $X$. This contradiction stems from the assumption that $X^*$ is not maximal. Hence, $X^*$ is the desired \mcs. The process continues recursively with $\beta_1$ and $\beta_2$.

It remains to argue that for all subprograms $\beta = \beta_1;X?;\beta_2$ we have that the set inclusion $X \supseteq \{\varphi\mid\Mubox{\beta'_1}\varphi \in l(\beta), \beta'_1\text{ instance of } \beta_1\} \cup \{\Mudiam{ \beta'_2} \psi \mid \psi \in r(\beta), \beta'_2 \text{ instance of } \beta_2 \}$ holds. For the first component, this is because we made sure the relevant $\varphi$ are in $X$ before proceeding to make the tests in $\beta_1$ maximal. Since $\vdash \Mubox{\alpha_1;\Phi?;\alpha_2}\psi$ entails $\Mubox{\alpha_1;\Phi'?;\alpha_2}\psi$ for $\Phi' \supseteq \Phi$, there is nothing left to prove.
For the second part, assume that there is an instance $\beta'_2$ of $\beta_2$ such that $\Mudiam{\beta'_2}\psi \notin X$ for some $\psi \in r(\beta)$. Since $X$ is an \mcs, $\Mubox{\beta'_2}\neg \psi \in X$. But then $\Transition{\Phi}{\alpha}{\Psi}$ is not even consistent, which contradicts the fact that all of the induction steps maintain consistency. In the case of a large loop, the same argument entails that all formulae required for a large loop to be maximally consistent are present at the tests.
\end{proof}
%%% Local Variables:
%%% mode: latex
%%% TeX-master: "pdl"
%%% End:

\subsection{Construction of the Canonical Model}
For the rest of this section, we assume, that all formulae are in the form defined in Lemma \ref{lem:program-normal-form}. Let $\tau'$ be $\tau$ extended with accessibility relation symbols $\alpha$ and $\alpha'$ for all \emph{large} $\pdocap[\tau]$-programs $\alpha$. We use $\alpha$-edges for \Forw-programs and $\alpha'$-edges for \Cyc-programs to better differentiate the two types and refer to the former as abstract forward edges and the latter as abstract loop edges.

%The two sets of edges of the form for large programs are used for programs in \Forw and \Cyc, we refer to the former as abstract forward edges and the latter as abstract loop edges.

Each point in the canonical model we construct is labeled by an $\pdocap[\tau]$-\mcs such that, after the construction is complete, a formula holds at a point if and only if it is in the \mcs that labels that point. The construction proceeds inductively. Each induction step consists of two stages: In the first stage, we add new points that witness the truth of diamond-type formulae $\langle \alpha \rangle \varphi$ at nodes of the previous induction step. These new states are connected to the previous ones via abstract $\beta$-edges from $\tau'\setminus \tau$, where $\beta$ is a suitable large program derived from $\alpha$. In the second stage these abstract $\beta$-edges are converted into subgraphs such that if there is an abstract $\beta$-edge from a node $u$ to a node $v$, then $u$ is connected via $\alpha$ to $v$. This is done by adding edges for abstract subprograms of $\beta$ and intermediate points, if necessary. The subgraph created for this is called the \emph{arena} witnessing that $\alpha$ connects $u$ to $v$. The whole process proceeds in a fashion such that no box-type formulae of the form $[\alpha]\varphi$ are violated.

The desired canonical model is the $\tau$-reduct of the limit of the inductive process. During this process, many points are labeled by the same \mcs. Since this can not be avoided (see Section~\ref{sec:canonical}), labels of points are more complex.

Let $\operatorname{mcs}(\pdocap)$ be the set of all maximally consistent \pdocap-sets. Then
\begin{align*}
\operatorname{Gen}_0 = &\, \operatorname{mcs}(\pdocap) \\
\operatorname{Gen}_{i+1} = &\, \{\arena(\Psi,\alpha,l)\mid \Psi \in \operatorname{mcs}(\pdocap), \alpha \in \Forw,l\in \operatorname{Gen}_i  \} \\
                        & \,\cup \{\arena(\alpha^\circlearrowleft,l)\mid\alpha \in \Forw,l\in \operatorname{Gen}_i\}
\end{align*}
with $\arena(\Psi,\alpha,l)$, $\arena(\alpha,l)$ and $\dom(l)$ defined inductively as follows:
\begin{itemize}
\item If $l \in \operatorname{Gen}_0$ then $\dom(l) = l$.
\item If not $\Mudiam{\alpha}\psi \in \dom(l)$ for all $\psi \in \Psi$, then $\arena(\Psi,\alpha,l)$ is empty.
\item If not $\Mudiam{\alpha^\circlearrowleft}\true \in \dom(l)$ then $\arena(\alpha^\circlearrowleft,l)$ is empty.
\item If $\Mudiam{\alpha}\varphi \in \dom(l)$ for all $\varphi$ in some \mcs $\Phi$, then there is a large program $\alpha'$ such that $\Transition{\Psi}{\alpha'}{\Phi}$ is maximally consistent. Then $\arena(\Phi,\alpha,l)$ is the nodes $l$, $r =(\Phi,\alpha,l)$ together with the abstract $\alpha'$-forward edge from $l$ to $r$ and the subgraph induced by it. Moreover, $\dom(r) = \Phi$.
\item If $\Mudiam{\alpha^\circlearrowleft}\true \in \dom(l)$ then there is an abstract $\alpha'$-loop such that $\Transition{\Psi}{\alpha'^\circlearrowleft}{\Psi}$ is maximally consistent. Then $\arena(\alpha^\circlearrowleft,l)$ is $l$ together with the abstract $\alpha'$-loop edge and the subgraph induced by it.
\end{itemize}
The subgraphs induced by abstract forward edges are again defined inductively:
\begin{itemize}
\item The subgraph induced by an abstract forward edge or an abstract loop edge $\alpha$ with $\alpha$ of the form $a$ for some atomic program $a \in \tau$ is an $a$-edge.
\item The subgraph induced by an abstract forward edge of the form $\alpha = \alpha_1;X?;\alpha_2$ from $u$ to $v$ consists of a node $w$ with $\dom(w) = X$, an abstract $\alpha_1$-forward edge from $u$ to $w$, an abstract $\alpha_2$-forward edge from $w$ to $v$ and the subgraphs induced by the abstract $\alpha_i$. Note that, by consistency of $\Transition{u}{\alpha}{v}$, we have $X\supseteq \{\varphi\mid\Mubox{\alpha'_1}\varphi \in \dom(u)\} \cup \{\Mudiam{\alpha'_2}\psi\mid \psi \in \dom(v)\}$ for all instances $\alpha'_1$ of $\alpha_1$ and $\alpha'_2$ of $\alpha_2$.
\item The subgraph induced by an abstract forward edge of the form $\alpha = \alpha_1\cap\alpha_2$ from $u$ to $v$ consists of abstract forward edges $\alpha_1$ and $\alpha_2$ from $u$ to $v$ and the subgraphs induced by them.
\end{itemize}
For abstract loop edges, the process is similar, but we annotate nodes with the programs needed to complete the loop in question and the set from which the loop starts. For an abstract $\alpha$-loop edge at $u$ set $l(u) = r(u) = \true?$.
\begin{itemize}
\item The subgraph induced by an abstract loop edge of the form $\alpha = \alpha_1;X?;\alpha_2$ from $u$ to $v$, which is part of a loop at $s$, consists of a node $w =(\lp(w),X,\rp(w),s)$ with $\dom(w) = X$, $\lp(w) = (\lp(u);\dom(u)?;\alpha_1)$, $\rp(w) = (\alpha_2;\dom(v)?;\rp(v))$, an abstract $\alpha_1$-loop edge from $u$ to $w$, an abstract $\alpha_2$-loop edge from $w$ to $v$ and the subgraphs induced by the abstract loop edges $\alpha_i$. Note that, by consistency of \Transition{u}{\alpha}{v} and by Lemma~\ref{lem:large-prog-max}, we have 
\begin{align*}
X\supseteq     &\{\varphi\mid\Mubox{\alpha'_1}\varphi \in \dom(u)\}\\
\cup &\{\Mudiam{\alpha'_2}\psi\mid \psi \in \dom(w)\} \\
\cup & \{\Mudiam{\big((\alpha'_2;\psi?;\beta'_2);\upsilon?;(\beta'_1;\varphi?;\alpha_1)\big)^\circlearrowleft}\true\mid \psi\in v, \upsilon\in s, \varphi\in u\}
\end{align*} 
for all instances $\alpha'_1,\alpha'_2,\beta_2,\beta_1$ of $\alpha_1,\alpha_2,l(v),r(v)$.
\item The subgraph induced by an abstract loop edge of the form $\alpha = \alpha_1\cap\alpha_2$ from $u$ to $v$ consists of abstract loop edges $\alpha_1$ and $\alpha_2$ from $u$ to $v$ and the subgraphs induced by them.
\end{itemize}
\begin{definition}
The canonical model for $\pdocap[\tau]$ is the $\tau$-reduct of the structure $\mathfrak{A} = \langle A, \{ p^\mathfrak{A} \}_{p \in \mathcal{P}}, \linebreak \{ R^\mathfrak{A} \}_{R \in \mathcal{R}} \rangle$, such that $A = \bigcup_{i=0}^\infty \operatorname{Gen}_i$,  $P^\mathfrak{A} = \{l \in A\mid P \in \dom(l) \}$ and $R^\mathfrak{A}$ as described above. The union is meant to be disjoint, with the exception that points $l \in \operatorname{Gen}_i$ and their counterparts in sets of the form $\arena(\Phi,\alpha,l)$ and $\arena(\alpha^\circlearrowleft,l)$ are identified.
\end{definition}
%The following follows from the construction of $\mathfrak{A}$:
\begin{lemma}
\label{lemma:boxesincanmodel}
\begin{enumerate}
\item \label{item:lemma:boxesincanmodel} If there is an abstract $\alpha$-forward edge from $u$ to $v$, then for all $[\alpha']\psi \in \dom(u)$ and $\alpha'$ an instance of $\alpha$, we have $\psi \in \dom(v)$.
\item \label{item:lemma:arenenastreelike} The set of arenas in $\mathfrak{A}$ decomposes $\mathfrak{A}$ into a forest-like structure. Any two arenas share at most one node, and any path from one arena to the other must go trough that node if it exists. For disjoint arenas, there is a node so that any path from one arena to the other must go through that node.
\item \label{item:lemma:cyclesknowntonodes} Any cycle in $\mathfrak{A}$ consists purely of nodes induced by abstract loop edges. Any $\alpha$-cycle at a node $u$ that consists purely of loop edges from the same arena is such that $\langle\alpha'^\circlearrowleft\rangle\true \in \dom(u)$ for any instance $\alpha'$ of $\alpha$.
\end{enumerate}
\end{lemma}
The proof is immediate from the construction of the model.

%%% Local Variables:
%%% mode: latex
%%% TeX-master: "pdl"
%%% End:

\subsection{Soundness and Completeness of the Canonical Model}
%It remains to prove that \ $\mathfrak{A}, l \models \varphi$ if $\varphi \in \dom(l)$.

\begin{lemma}[Gateway Lemma]
\label{lem:gateway}
Let $\alpha \in \Forw$ be a program which contains at least one occurrence of the $;$-operator. Let $\mathfrak{A}$ be a  Kripke structure and let $u$ and $w$ be not necessarily distinct worlds in $\mathfrak{A}$ such that $\mathfrak{A}$ is a minimal witness graph for $\Transition{u}{\alpha}{v}$. Let $v$ be a node different from $u$ and $w$ such that all paths of length $> 0$ from $u$ to $w$ must go through $v$.
Then there is $\beta = \beta_1;\beta_2$ such that $\vdash \beta\Rightarrow\alpha$ and $\Transition{u}{\beta_1}{w}$ and $\Transition{w}{\beta_2}{v}$.
\end{lemma}
\begin{proof}
The proof is by induction on the construction of $\alpha$. Since $\alpha$ contains at least one occurrence of the $;$-operator, $\alpha \not = \bigcap_{i \in	I} R_i$ for a finite set of $R_i \in \mathcal{R}$. Hence, the base case is that $\alpha = \alpha_1;\alpha_2$ and $\Transition{u}{\alpha_2}{v}$ and $\Transition{v}{\alpha_2}{w}$. 

If $\alpha = \alpha_1;\alpha_2$ and not $\Transition{u}{\alpha_1}{v}$ or not $\Transition{v}{\alpha_2}{w}$, then there is $v'$ such that $\Transition{u}{\alpha_1}{v'}$ and $\Transition{v'}{\alpha_2}{w}$. There are two cases: All paths from $u$ to $v'$ go through $v$, or all paths from $v'$ to $w$ go through $v$. Otherwise, there is a path from $u$ to $w$ via $v'$ that does not go through $v$. If all paths from $u$ to $v'$ go through $v$, by the induction hypothesis there is $\alpha'_1;\alpha''_1$ such that $\Transition{u}{\alpha'_1}{v}$ and $\Transition{v}{\alpha''_1}{v'}$. Then $\beta_1 = \alpha'_1$ and $\beta_2 = \alpha''_1;\alpha_2$ are as desired. If all paths from $v'$ to $w$ go through $v'$, an application of the induction hypothesis yields $\alpha'_2$ and $\alpha''_2$ such that $\Transition{u}{\alpha_1;\alpha'_2}{v}$ and $\Transition{v}{\alpha''_2}{w}$. In both cases, clearly $\vdash [\alpha]p \Rightarrow [\beta]p$.

If $\alpha = \alpha_1 \cap \alpha_2$, then by the induction hypothesis there are $\alpha'_1;\alpha''_1$ and $\alpha'_2;\alpha''_2$ such that $\Transition{u}{\alpha'_i}{w}$ and $\Transition{v}{\alpha''_i}{v}$ for $i = 1,2$. Then, via Axiom~\eqref{ax:semicolon}, we have $\vdash (\alpha'_1 \cap \alpha'_2);(\alpha''_1 \cap \alpha''_2) \Rightarrow \alpha_1 \cap \alpha_2$.

\end{proof}

\begin{lemma}
\label{lem:articulation-point-of-set}
Let $\alpha$ be a program. Let $\mathfrak{A}$ be a  Kripke structure and let $u$ and $w$ be not necessarily distinct worlds in $\mathfrak{A}$ such that $\mathfrak{A}$ is a minimal witness graph for $\Transition{u}{\alpha}{w}$. Let $v$ be a node and $E$ a nonempty set of edges such that all paths from $u$ to $w$ are such that $v$ occurs before and after any edge $e \in E$. 

Then there are programs $\alpha'$ and $\beta$ such that $\alpha'$ contains an occurrence of $\Mudiam{\beta^\circlearrowleft}\true?$, $\vdash \alpha' \Rightarrow \alpha$ and the witness graph for $\alpha''$ contains no node from $X$.
\end{lemma}
\begin{proof}
By application of Lemma~\ref{lem:gateway} and the fact that $\vdash \beta^\circlearrowleft \Rightarrow \beta$.
\end{proof}

\begin{lemma}
\label{lem:union-of-abstract-edges}
	Let $u, v \in A$ and $\alpha_1, \alpha_2$ be large programs. If there are abstract $R_{\alpha_1}$- and abstract $R_{\alpha_2}$-edges between $u$ and $v$ either both forward or loop edges, then either $\alpha_1 = \alpha_2$, or there is a large program $\alpha$ of the same kind such that $\alpha \geq \alpha_1$ and $\alpha \geq \alpha_2$, and there is an abstract $R_\alpha$-edge between $u$ and $v$.
\end{lemma}
\begin{proof}
By the construction of arenas. Each abstract edge $\alpha$ that is not that inducing an arena itself has a parent edge $\alpha'$ in the inductive process such that the top operator of $\alpha'$ is either $;$ or $\cap$. Usage of the $;$-operator changes either source or sink of an abstract edge, and for different programs intermediate points are different. Since $\alpha_1$ and $\alpha_2$ share the same source and sink nodes, they must have a common parent solely via $\cap$-operators.
\end{proof}

\begin{lemma}
\label{lem:abstract-edges-common-art-point}
Let $u, v,w \in A$ and $\alpha_1, \alpha_2$ be large programs. If there is an abstract $\alpha_1$-edge from $u$ to $w$ and an abstract $\alpha_2$-edge from $w$ to $v$, then either there is an abstract $\alpha$-edge from $u$ to $v$ and $\alpha \geq \alpha_1;\alpha_2$, or all paths from $u$ to $v$ contain $w$. 
\end{lemma}
\begin{proof}
By construction of arenas. If $u$ and $v$ are in different arenas the claim follows from Item~\eqref{item:lemma:arenenastreelike} of Lemma~\ref{lemma:boxesincanmodel}. Otherwise, each abstract edge $\alpha$ that is not inducing an arena itself has a parent edge $\alpha'$ in the inductive process such that the top operator of $\alpha'$ is either $;$ or $\cap$. If there is no common parent of $\alpha_1$ and $\alpha_2$ from $u$ to $v$, then at least one of the $\alpha_i$ must have a program with top operator $;$ as its most common parent. By the construction of arenas, all paths from $u$ to $v$ must go through $w$.
\end{proof}

\begin{lemma}
\label{lem:box}
	For all points $u,v \in A$ with and all programs $\alpha$ such that $\Transition{u}{\alpha}{v}$ and all formulae $\psi$ it holds that: If $[\alpha]\psi \in dom(u)$, then $\psi \in dom(v)$.
\end{lemma}
\begin{proof}
If $u = v$, because $\vdash \Mubox{\alpha}\psi \rightarrow \Mubox{\alpha^\circlearrowleft}\psi$, we can invoke Lemma~\ref{lem:boxes-circle}. So without loss of generality, $u \not = v$.

Since $u \not = v$, we know that $\alpha \in \Forw$.  Moreover, we can assume that both $u$ and $v$ and the witness graph for \Transition{u}{\alpha}{v} are both contained in the same arena. Otherwise, $u$ and $v$ are in different arenas and, by Item~\eqref{item:lemma:arenenastreelike} of Lemma~\ref{lemma:boxesincanmodel}, there is $w$ such that all paths from $u$ to $v$ must go through $w$. By the Gateway Lemma~\ref{lem:gateway}, $\alpha$ can be rewritten as $\alpha_1;\alpha_2$ such that $\Transition{u}{\alpha_1}{w}$ and $\Transition{w}{\alpha_2}{v}$. The claim of the lemma reduces to prove that $\psi' = \Mubox{\alpha_2}\psi \in \dom(w)$. In a similar fashion, if $u$ and $v$ are in the same arena, but parts of the witness graph of $\Transition{u}{\alpha}{v}$ are not, the conditions of Lemma~\ref{lem:articulation-point-of-set} are met and parts in different arenas can be reduced to a test.

We prove that there is a sequence $\beta_1;\dotsc;\beta_n$ of abstract edges such that there are instances $\beta'_i$ of the $\beta_i$ with $\vdash \beta'_1;\dotsc;\beta'_n \Rightarrow \alpha$. By Item~\eqref{item:lemma:boxesincanmodel} of Lemma~\ref{lemma:boxesincanmodel}, this proves the lemma.
We prove this by induction over the structure of $\alpha$. If $\alpha = a$ for some atomic program $a$, then by construction, there is an abstract $a$-edge from $u$ to $v$. If $\alpha = \alpha_1;\alpha_2$, there is $w$ such that $\Transition{u}{\alpha_1}{w}$ and $\Transition{w}{\alpha_2}{v}$. By the induction hypothesis there are sequences $\beta^1_1;\dotsc;\beta^1_n$ and $\beta^2_1;\dotsc;\beta^2_m$ and points $u =w^1_0,\dotsc,w^1_n = w$ and $w = w^2_0;\dotsc;w^2_m = v$ such that there are abstract $\beta^j_i$-edges from $w^j_{i-1}$ to $w^j_{i}$ and there are instances $\beta^j_i$ of the $\beta^j_i$ such that $\vdash \beta^1_1;\dotsc;\beta^1_n \Rightarrow \alpha_1$ and $\vdash \beta^2_1;\dotsc;\beta^2_m \Rightarrow \alpha_2$. Then $\vdash \beta^1_1;\dotsc;\beta^1_n;\beta^2_1;\dotsc;\beta^2_m \Rightarrow \alpha_1;\alpha_2$.

If $\alpha = \alpha_1\cap\alpha_2$, then $\Transition{u}{\alpha_1}{v}$ and $\Transition{u}{\alpha_2}{v}$. By the induction hypothesis, there are sequences $\beta^1_1;\dotsc;\beta^1_n$ and $\beta^2_1;\dotsc;\beta^2_m$ and points $u =w^1_0,\dotsc,w^1_n = v$ and $u = w^2_0;\dotsc;w^2_m = v$ such that there are abstract $\beta^j_i$-edges from $w^j_{i-1}$ to $w^j_{i}$ and there are instances $\beta^j_i$ of the $\beta^j_i$ such that $\vdash \beta^1_1;\dotsc;\beta^1_n \Rightarrow \alpha_1$ and $\vdash \beta^2_1;\dotsc;\beta^2_m \Rightarrow \alpha_2$. There are two cases: If both sequences have length one, we can apply Lemma~\ref{lem:union-of-abstract-edges} to obtain an abstract edge $\beta$ from $u$ to $v$ such that there are instances $\beta_1$ and $\beta_2$ of $\beta$ with $\vdash \beta_1 \Rightarrow \alpha_1$ and $\vdash \beta_2 \Rightarrow \alpha_2$. Then the sequence just consisting of $\beta$ is as desired.

If at least one sequence has length longer than one, we begin replacing subsequences of the form $\beta^j_i;\beta^j_{i+1}$ by abstract $\beta^j_i;\beta^j_{i+1}$-edges, if possible. Either both sequences reach length one, and we can apply the previous case, or the conditions of Lemma~\ref{lem:abstract-edges-common-art-point} apply and there is $w$ such that all paths from $u$ to $v$ go through $w$. By the Gateway Lemma~\ref{lem:gateway}, we can rewrite $\alpha_1$ and $\alpha_2$ to $\alpha^1_1;\alpha^2_1$ and $\alpha^1_2;\alpha^2_2$ such that $\vdash \alpha^1_i;\alpha^2_i \Rightarrow \alpha_i$ and $\Transition{u}{\alpha^1_i}{w}$ and $\Transition{w}{\alpha^2_i}{v}$ for $i =1,2$. Then $\vdash (\alpha^1_1\cap\alpha^1_2);(\alpha^2_1\cap\alpha^2_2) \Rightarrow (\alpha_1^1;\alpha^2_1)\cap(\alpha^1_2;\alpha^2_2)$ and, by the induction hypothesis, there are $\gamma^1_1;\dotsc;\gamma^1_n$ and $\gamma^2_1;\dotsc;\gamma^2_m$ such that $\gamma'^1_1;\dotsc;\gamma'^1_n;\gamma'^2_1;\dotsc;\gamma'^2 \Rightarrow (\alpha^1_1\cap\alpha^1_2);(\alpha^2_1\cap\alpha^2_2)$ for some instances $\gamma'^j_i$ of $\gamma^j_i$. This finishes the proof.
\end{proof}

\begin{lemma}
\label{lem:boxes-circle}
For all points $u \in A$ and all programs $\alpha$ such that $\Transition{u}{\alpha^\circlearrowleft}{u}$ and all formulae $\psi$ it holds that: If $\Mubox{\alpha^\circlearrowleft}\psi \in dom(u)$, then $\psi \in dom(u)$.
\end{lemma}
\begin{proof}
The proof proceeds similarly to Lemma~\ref{lem:box}. Again, without loss of generality, the witness graph for $\Transition{u}{\alpha^\circlearrowleft}{u}$ is contained in one arena. Instead of converting the program into abstract forward edges, we convert it into abstract loop edges. Once this is done, we invoke Item~\eqref{item:lemma:cyclesknowntonodes} of Lemma~\ref{lemma:boxesincanmodel}.
\end{proof}

\begin{lemma}[Existence Lemma]
\label{lem:existence}
	For any $u \in A$ and any program $\alpha$, if $\langle \alpha \rangle \varphi \in dom(u)$, then there is a state $v \in A$ with $\Transition{u}{\alpha}{v}$, such that $\varphi \in dom(v)$.
\end{lemma}
\begin{proof}
If $\alpha\in \Cyc$, there is an $\alpha$-loop at $u$ by the construction of $\mathfrak{A}$.
	
The other case is that $\alpha\in \Forw$. Set $X^* = \{ \varphi \mid \Mudiam{\alpha}\varphi \in dom(u) \}$. This set, in general, is inconsistent. However, the set of nonempty subsets of $X$ that contain $\varphi$ and all $\psi$ such that $\Mubox{\alpha}\psi \in \dom(u)$ is nonempty and satisfies the conditions of Zorn's Lemma. Hence, there is a maximal such set $\Phi$. We claim that it is an \mcs: For each $\pdocap[\tau]$-formula $\psi$, either $\Mudiam{\alpha}\psi \in \dom(u)$ or $\Mubox{\alpha}\neg \psi \in \dom(u)$. In the latter case, $\psi \in \Phi$. In the former case, both $\Mudiam{\alpha}\psi \in \dom(u)$ and $\Mudiam{\alpha}\psi \in \dom(u)$. If neither of these is in $\Phi$, then $\Phi$ is not maximal since $\vdash \neg\Mudiam{\alpha}\psi \vee \neg \Mudiam{\alpha}\psi \leftrightarrow \Mubox{\alpha}\false$. So $\Phi$ was maximal after all.
	
Further, let $v = (\Psi, \alpha, u)$. By construction, $\varphi \in dom(v)$ and there is a large program $\alpha'$ such that $\Transition{u}{\alpha'}{v}$ is maximally consistent. Thus $\Transition{u}{\alpha}{v}$  since there is an $\alpha'$-edge from $u$ to $v$.
\end{proof}

\begin{lemma}
For all $v \in A$ and for all $\varphi \in \pdocap$ we have $\mathfrak{A},v\models \varphi$ if and only if $\varphi \in dom(v)$.
\end{lemma}
\begin{proof}
We only prove the if part. The other direction follows from contraposition and the fact that an \mcs contains every formula or its negation.
The case for atomic propositions  is by the definition of the valuations in the canonical model, the case for boolean connectives follows from the closure properties of \mcs.  The case for $\varphi = \Mudiam{\alpha} \psi$ follows from Lemma~\ref{lem:existence}. This leaves the case $\varphi = \Mubox{\alpha} \psi$, which follows from Lemma~\ref{lem:box} and Lemma~\ref{lem:boxes-circle}.
\end{proof}

\begin{theorem}
$\vdash$ is sound and complete: $\Phi \vdash \varphi$ if and only if $\Phi \models \varphi$.
\end{theorem}
\begin{proof}
Soundness is by Lemma~\ref{lem:vdashsound}. For the sake of contradiction, assume that $\Phi \models \varphi$ but $\Phi \not\vdash \varphi$. Then $\Phi \cup \{\neg \varphi\}$ is consistent and contained in an \mcs $\Phi'$. But the generation $0$ node $\mathfrak{A},\Phi'$ is such that $\mathfrak{A},\Phi' \models \psi$ if and only if $\psi \in \Phi'$. Since $\neg\varphi \in \Phi'$ and $\Phi \subseteq \Phi'$, this contradicts $\Phi \models \varphi$. Hence, $\vdash$ is complete.
\end{proof}
% \begin{corollary}
% \pdocap is compact.
% \end{corollary}
% \begin{proof}
% If $\Phi$ is finitely satisfiable, it must be consistent. But then $\Phi$ is contained in an \mcs $\Phi'$ and $\mathfrak{A},\Phi'$ is a model of $\Phi'$ and, hence, of $\Phi$.
% \end{proof}
%%% Local Variables:
%%% mode: latex
%%% TeX-master: "pdl"
%%% End:

% !TEX root =  pdl.tex

\section{Conclusion}

We have presented a refined construction of a canonical model for the iteration-free fragment
of Propositional Dynamic Logic with Intersection and Tests (\pdocap) and used this to prove completeness of an axiom system
for this logic. The trick that handles the combinatorial difficulties introduced by the interaction
between the intersection operator and test programs is the use of several copies of a maximally consistent set
for worlds in this Kripke model. 

As in turns out, there are parallels between our construction and that in \cite{Ba2003.6}, respectively those used for fragments of \pdl, e.g.~in
\cite{journals/jancl/BalbianiC98}. Both constructions use multiple copies of maximally consistent sets in their canonical model, and both construct this
model as the countable union of partial approximations, each of which is generated by constructing witnesses for all diamond formulae that lack such a
witness. The construction in this paper is more explicit and more constructive, for example because it does not rely on a well-ordering of unsatisfied
diamond formulae, or the language being countable.

The rules of the proof calculi in this paper and in \cite{Ba2003.6} also have similarities. However, our approach does not rely on computable or
recursively enumerable auxiliary functions but rather incorporates their content into the calculus itself. % We were also able to show strong completeness of the calculus presented here, something which is not claimed in \cite{Ba2003.6}.

Future work will attempt to derive a weakly complete axiomatisation for full \pdtcap, i.e.\ the logic including Kleene
iteration, based on such a refined canonical model construction over finite sets of formulae.    

\paragraph*{Acknowledgment.} We thank Philippe Balbiani for the discussion we had on the topic of axiomatisations of Iteration-Free
PDL with Intersection.

%%% Local Variables:
%%% mode: latex
%%% TeX-master: "pdl"
%%% End:

%% Bibliography
%% Make sure to use the bibliographystyle aiml16.
\bibliographystyle{eptcs}
\bibliography{../literature}

%\Appendix

\end{document}